\documentclass{article}
\usepackage{spconf}
\usepackage{flushend}  
\usepackage{blindtext, graphicx}
\usepackage{bm}
\usepackage{booktabs}
\usepackage{subfig}
\usepackage[font=footnotesize]{caption}
\usepackage{hyperref}
\usepackage{algorithm2e}
\usepackage{siunitx}

\usepackage{widebar}
\usepackage{localmath}

\providecommand{\SINR}{\ensuremath{\mathsf{SINR}}}

\graphicspath{{figures/}{figures/auxive_sim_20191001-224417_f0343b99ae/}}

\begin{document}
\ninept

\title{Fast Independent Vector Extraction by Iterative SINR Maximization}

\name{Robin Scheibler and Nobutaka Ono%
\thanks{
This research was supported by JSPS KAKENHI Grant Numbers 17F17049, 16H01735 and JST CREST Grant Number JPMJCR19A3.
}%
\thanks{
Code and data to reproduce the results of this paper are available at \protect\url{https://github.com/onolab-tmu/code_2020ICASSP_five}.%
}}
\address{Tokyo Metropolitan University, Tokyo, Japan}

%

\maketitle

\begin{abstract}
	We propose fast independent vector extraction (FIVE), a new algorithm that blindly extracts a single non-Gaussian source from a Gaussian background.
  The algorithm iteratively computes beamforming weights maximizing the signal-to-interference-and-noise ratio for an approximate noise covariance matrix.
  We demonstrate that this procedure minimizes the negative log-likelihood of the input data according to a well-defined probabilistic model.
  The minimization is carried out via the auxiliary function technique whereas, unlike related methods, the auxiliary function is globally minimized at every iteration.
	Numerical experiments are carried out to assess the performance of FIVE.
  We find that it is vastly superior to competing methods in terms of convergence speed, and has high potential for real-time applications.
\end{abstract}

\section{Introduction}

Blind source extraction (BSE) aims at separating a single target signal from background noise without any prior information~\cite{Statistics:1985tc,Cardoso:1993ii}.
While BSE predates independent component analysis (ICA)~\cite{Comon:1512057}, the two problems are tightly related~\cite{Amari:1998cz,CrucesAlvarez:2004hq,Javidi:2010if}.
It can be seen as a blind source separation (BSS) problem where only one source is retrieved.
BSE most often relies on the independence of a non-Gaussian source contrasted to a Gaussian~\cite{Koldovsky:fn}, or non-Gaussian~\cite{Koldovsky:2018bw}, background.

Our focus is on audio applications where the mixture is typically convolutive.
In this case, the frequency domain BSS framework \cite{Smaragdis:1998kl} can be leveraged to transform the convolution into point-wise multiplication in the time-frequency domain.
Then, BSE can be applied in parallel to all the narrowband channels.
Done directly, however, this may lead to a problem whereas different sources are extracted at different frequencies.
This is known as the permutation ambiguity problem in ICA~\cite{Sawada:fk}.
An elegant solution is to consider multivariate probability distributions over frequencies, giving rise to the independent vector analysis (IVA)~\cite{Hiroe:2006ib,Kim:2006ex} and extraction (IVE)~\cite{Koldovsky:fn} paradigms in BSS and BSE, respectively.
For BSE, OGIVE, a gradient-ascent algorithm, was proposed and shown to be effective~\cite{Koldovsky:fn}.
However, the speed and convergence guarantees of gradient methods are limited.
Recently, we introduced OverIVA, an algorithm for overdetermined BSS with fast and guaranteed convergence.
OverIVA assumes a super-Gaussian model for sources and a Gaussian background.
The algorithm is obtained by applying the auxiliary function optimization method, similarly to AuxIVA~\cite{Ono:2011tn}, with the addition of orthogonality constraints between the target and background signals~\cite{Cardoso:1994wj}.
This algorithm is applicable to IVE.

In this paper, we introduce a new algorithm for fast IVE (FIVE).
FIVE iteratively applies maximum signal-to-interference-and-noise ratio ($\SINR$) beamforming~\cite{VanTrees:2002ec} to improve upon an initial estimate of the target signal.
We show that this deceptively simple algorithm can be rigorously derived by minimizing the same cost function as OverIVA.
Whereas OverIVA relied on orthogonal constraints for the separation of the background, FIVE solves the minimization of the auxiliary function to its global minimum.
Interestingly, this is done by solving exactly a special case of the hybrid exact-approximate diagonalization (HEAD) problem~\cite{Yeredor:hr}.
Experiments reveal that FIVE is blazingly fast and only requires a few iterations to achieve over \SI{4}{\decibel} signal-to-distortion ratio (SDR) improvement.
Further investigation reveals that FIVE behaves similarly to OverIVA and OGIVE in the presence of a mismatched background.

The reminder of this paper is organized as follows.
\sref{background} introduces notation, signal model and maximum $\SINR$ beamforming.
FIVE is described informally in \sref{algorithm} and analyzed in \sref{analysis}.
The experiments are discussed in \sref{experiments}. \sref{conclusion} concludes.

\section{Background}
\seclabel{background}

\subsection{Signal Model and Notation}

We consider the mixture of a single source with an arbitrary background noise recorded by $M$ microphones.
In the time-frequency domain, our signal model is
\begin{equation}
  x_{mfn} = s_{mfn} + b_{mfn},
\end{equation}
where $x_{mfn}$, $s_{mfn}$, and $b_{mfn}$ are the short-time Fourier transforms (STFT)~\cite{Allen:1977in} of the microphone, target, and background signals, respectively.
The indices $f=1,\ldots,F$ and $n=1,\ldots,N$ are the discrete frequency bins and frames, respectively.
For convenience, we can group all microphone signals in the complex-valued vectors
\begin{equation}
	\vx_{fn} = \begin{bmatrix} x_{1fn} & \cdots & x_{Mfn} \end{bmatrix} \in \C^M.
\end{equation}
In the rest of the manuscript, we use lower and upper case bold letters for vectors and matrices, respectively.
Furthermore, $\mA^\top$, $\mA^\H$, $\det(\mA)$ and $\trace(\mA)$ denote the transpose, conjugate transpose, determinant and trace of matrix $\mA$, respectively.
Unless specified otherwise, indices $m$, $f$, and $n$ always take the ranges defined here.

\subsection{Maximum SINR Beamforming}

Beamforming addresses the problem of finding the optimal weight vectors $\vw_f$ to combine the microphone signals such that $\vw_f^\H \vx_{fn}$ is close to the target signals.
This is generally achieved by optimizing a well-chosen cost function.
One can define the signal-to-interference-and-noise ratio ($\SINR$)
\begin{equation}
	\SINR_f[\vw] = \frac{\vw^\H \mSigma^{(s)}_f \vw}{\vw^\H \mSigma^{(b)}_f \vw}
	\elabel{sinr}
\end{equation}
where $\mSigma^{(s)}_f$ and $\mSigma^{(b)}_f$ are the covariance matrices of target and background signals, respectively.
The Maximum $\SINR$ beamformer is the one maximizing \eref{sinr} \cite{VanTrees:2002ec}.
Note that, if the target and background are uncorrelated, replacing $\mSigma^{(s)}_f$ by $\mSigma^{(x)}_f$, the covariance matrix of the microphone signals, in \eref{sinr} only changes the ratio by a constant additive factor.

In practice, we must approximate the covariance matrices from the available data,
e.g. the estimate of $\mSigma^{(x)}_f$ is the sample covariance matrix of the input data
\begin{equation}
  \mC_f = \frac{1}{N} \sum_n \vx_{fn} \vx_{fn}^\H.
  \elabel{covmat}
\end{equation}
Now, provided a (possibly scaled) estimate $\mV_f \sim \mSigma_f^{(b)}$, we can obtain $\vw_f$ as the result of the following optimization,
\begin{equation}
	\vw_f = \underset{\vw \in \C^M}{\arg\max}\ \frac{\vw^\H \mC_f \vw}{\vw^\H \mV_f \vw} \approx \underset{\vw \in \C^M}{\arg\max}\ C_1 \SINR_f[\vw] + C_2,
	\elabel{gev_problem}
\end{equation}
where $C_1 > 0$ and $C_2$ are arbitrary constants.
The optimizer of \eref{gev_problem} is the generalized eigenvector corresponding to the largest generalized eigenvalue for the problem $\mC_f \vw = \lambda \mV_f \vw$.
However, finding a good estimate $\mV_f$ turns out to be a challenging problem, and the Maximum $\SINR$ beamformer is difficult to use in practice.

\section{Algorithm}
\seclabel{algorithm}

Suppose we are given an initial guess $\hat{s}_{fn}$ of the target signal, typically one of the microphone signals.
Then an (unscaled) estimate of the background covariance matrix is
\begin{equation}
  \mV_f = \frac{1}{N} \sum_n \varphi_n(r_n) \vx_{fn} \vx_{fn}^\H, \quad \forall f,
  \elabel{weighted_covmat}
\end{equation}
where $\varphi_n(r)\,:\,\R_+ \to \R$ is a, yet-to-be-defined, strictly decreasing function, and $r_n$
is the magnitude of the target signal estimate,
\begin{equation}
  r_{n} = \sqrt{\sum\nolimits_f |\hat{s}_{fn}|^2}, \quad \forall n.
  \elabel{activity}
\end{equation}
Due to $\varphi_n(r_n)$, the importance of target dominated frames, i.e. where $r_n$ is large, is reduced, while background dominated frames are emphasized.
We can now compute $\vw_f$ as in \eref{gev_problem} by solving a generalized eigenvalue problem.
Now, using the newly obtained demixing filter $\vw_f$, we update the target signal estimate
\begin{equation}
  \hat{s}_{fn} \gets \vw_{f}^\H \vx_{fn}, \quad \forall f, n.
\end{equation}
The procedure is then repeated until convergence, or for a fixed number of iterations.
Note that a normalization step is needed to keep the extracted signal scale under control.

A simple improvement to this algorithm is to pre-whiten the input signal so that $\mC_f = \mI$.
Then, solving \eref{gev_problem} only requires the computation of the smallest eigenvalue and corresponding eigenvector of $\mV_f$.
Pseudo-code for the final form of the algorithm is provided in~\algref{five}.
While the procedure just presented might seem ad-hoc, we show in the following section that it can be rigorously derived from the minimization of a well-chosen cost function, and that its convergence is, in fact, guaranteed.

\begin{algorithm}[t]
\SetKwInOut{Input}{Input}\SetKwInOut{Output}{Output}
\Input{Input signals $\{ \vx_{fn} \}$, Initial estimate $\{\hat{s}_{fn}\}$}
\Output{Extracted signal $\{ s_{fn} \}$ }
\DontPrintSemicolon
\# Input pre-whitening\;
$\tilde{\vx}_{fn} = \mQ^{-\H}_f \vx_{fn}$, with $\frac{1}{N} \sum_n \vx_{fn} \vx_{fn}^\H =\mQ_f^\H\mQ_f$, $\forall f$\;
\For{loop $\leftarrow 1$ \KwTo $\text{max. iterations}$}{
  $r_{n} \gets \sqrt{\sum_f |\hat{s}_{fn}|^2},\ \forall n$\;
  \For{$f \gets 1$ \KwTo $F$}{
    $\wt{\mV}_{f} \gets \frac{1}{N} \sum_n \varphi_n(r_n) \tilde{\vx}_{fn} \tilde{\vx}_{fn}^\H$\;
    Let $\lambda_M$ and $\vr_M$ be the smallest eigenvalue of $\wt{\mV}_f$ and corresponding eigenvector, respectively\;
    $\vw_f \gets \lambda_M^{-\frac{1}{2}} \vr_M$\;
    $\hat{s}_{fn} \gets \vw_{f}^\H \tilde{\vx}_{fn},\ \forall n$\;
  }
}
\caption{FIVE: Fast Independent Vector Extraction}
\label{alg:five}
\end{algorithm}

\section{Derivation}
\seclabel{analysis}

We turn now to the analysis of the derivation of \algref{five}.
We prove that it extracts the maximum likelihood source estimate under a well-defined probabilistic model.
We consider the source extraction problem as a special case of determined blind source separation.
More specifically, we want to find the $M\times M$ demixing matrix
\begin{equation}
	\mW_f = \begin{bmatrix} \vw_f & \mJ_f \end{bmatrix}^\H,
\end{equation}
where $\mJ_f \in \C^{M\times M-1}$ is such that $s_{fn} = \vw_f^\H \vx_{fn}$ is independent from $\vz_{fn} = \mJ_f^\H\vx_{fn}$.

\begin{theorem}
  Let the three following assumptions hold.
  \begin{itemize}
    \item Target signal and background are statistically independent.
    \item The source signal distribution is spherical super-Gaussian
      \begin{equation}
        p_{S_n}(s_{1n}, \ldots, s_{Fn}) \sim e^{-G_n\left(\sqrt{\sum_f s_{fn}}\right)},
      \end{equation}
      with $G_n\::\:\R_+ \to \R$, strictly increasing, differentiable, and such that $G_n^\prime(r)/r$ is strictly decreasing (see \cite{Ono:2011tn,Ono:2010hh} for details).
    \item The background is Gaussian with arbitrary covariance structure across channels, but uncorrelated over frequencies.
  \end{itemize}
	Then, \algref{five} with
	\begin{equation}
		\varphi_n(r) = \frac{G^{\prime}_n(r)}{2 r}
		\elabel{varphi_def}
	\end{equation}
	is guaranteed to converge to a stationary point of the negative log-likelihood of the observed signal.
  \label{theorem:1}
\end{theorem}

The rest of this section proves the theorem.
Based on the probabilistic model enounced in Theorem~\ref{theorem:1}, we can write explicitly the negative log-likelihood of the observed signal,
\begin{multline}
  \calL = -2N\sum_f \log|\det(\mW_f)| + \sum_{n} G_n\left(\sqrt{\sum\nolimits_f |\vw_f^\H\vx_{fn}|^2 }\right) \\
  + \sum_{fn} \vx_{fn}^\H \mJ_f \mB^{-1}_f \mJ_f^\H \vx_{fn},
  \elabel{cost_function}
\end{multline}
where $\mB = \Expect{\vz_{fn}\vz_{fn}^\H}$ is the covariance matrix of the background after demixing.
At this point, we will assume that $\mB$ is known or can be estimated.
As we will find out, it is in fact irrelevant.
The maximum likelihood estimate of the target signal is provided by minimizing \eref{cost_function} with respect to $\vw_f$ and $\mJ_f$.
While direct minimization of $\calL$ is hard due to the non-quadratic term in $\vw_f$, it can be done via the auxiliary function approach~\cite{Ono:2010hh,Ono:2011tn}.
We make use of an inequality for super-Gaussian sources to create a majorizing function of~\eref{cost_function}.
\begin{lemma}[from \cite{Ono:2010hh}]
  Let $G_n(r)$ be as defined in Theorem~\ref{theorem:1}. Then,
  \begin{equation}
  G_n(r) \leq G_n^\prime(r_0)\frac{r^2}{2 r_0} + \left( G_n(r_0) - \frac{r_0}{2} G_n^\prime(r_0)\right),
  \end{equation}
  with equality for $r = r_0$.
\end{lemma}
Then, the majorizing function $\calL_2$ is as follows
\begin{multline}
  \calL \leq \calL_2 = -2N\sum_f \log|\det(\mW_f)| + N \sum_f \vw_f^\H \mV_f \vw_f \\
  + N \sum_{f} \trace\left(\mJ_f^\H \mC_f \mJ_f \mB^{-1}_f\right) + \text{constant},
  \elabel{majorizer}
\end{multline}
with $\mV_f$, $\mC_f$, and $r_n$ defined in~\eref{weighted_covmat}, \eref{covmat}, and~\eref{activity}, respectively.
Then, the auxiliary function method (also known as majorization-minimization) consists in 
iteratively minimizing~\eref{majorizer} and recomputing $r_n$ based on the new demixing filter.
This method is guaranteed to converge to a stationary point of~\eref{cost_function}~\cite{Lange:2016wp}.
%
Equating the gradient of $\calL_2$ to zero leads to the following quadratic system of equations
\begin{equation}
  \begin{bmatrix} \vw_f^\H \\ \mJ_f^\H \end{bmatrix}
  \begin{bmatrix} \mV_f \vw_f & \mC_f \mJ_f \end{bmatrix}
  =
  \begin{bmatrix} 1 & \vzero^\top \\ \vzero & \mB_f \end{bmatrix}, \quad \forall f.
  \elabel{mod_head}
\end{equation}
We omit the index $f$ from here on for convenience.
This is a special case of the HEAD problem \cite{Yeredor:hr} where $M-1$ columns share the same matrix.
Although a general closed form solution of HEAD is unknown for $M > 2$, we show that~\eref{mod_head} can be solved exactly.
This is a generalization of the case $M=2$, presented in~\cite{Ono:2012wa}.
\providecommand{\mRbar}{\boldsymbol{\widebar{R}}}
\begin{proposition}
  Let $\mC$ and $\mB$, both Hermitian matrices, have decompositions $\mC = \mQ^\H \mQ$ and  $\mB = \mU^\H \mU$,
  and let $\lambda_1 \geq \ldots \geq \lambda_M$ and $\vr_1, \ldots, \vr_M$ be the eigenvalues and eigenvectors, respectively, of $\wt{\mV} = \mQ^{-\H} \mV \mQ^{-1}$.
  In addition, let $\mRbar_k$ be the $M\times M-1$ matrix whose columns are $\vr_{\ell}$, $\forall \ell \neq k$.
  Then, 
  \begin{equation}
    \vw = \frac{1}{\sqrt{\lambda_k}} \mQ^{-1} \vr_k,\quad \mJ = \mQ^{-1} \mRbar_k \mU, \elabel{head_sol}
  \end{equation}
  is a solution to~\eref{mod_head} for every $k=1,\ldots,M$.
\end{proposition}
\begin{proof}
The proof follows by substituting \eref{head_sol} into \eref{mod_head}, applying the properties of the eigenvectors, and the decompositions.
\end{proof}
As we have just shown, there are $M$, possibly distinct, solutions to~\eref{mod_head}, corresponding to $M$ stationary points of~\eref{majorizer}.
\begin{proposition}
The global minimum of~\eref{majorizer} is given by the minimum eigenvalue, i.e. $k=M$.
\end{proposition}
\begin{proof}
  Under the choice \eref{head_sol}, the only non-constant term in \eref{majorizer} is the log-determinant.
  All we need to show is that the determinant is maximized.
	Because $\mQ$ and $\mU$ are independent of $k$, and
	\begin{equation}
		\mW^\H = \mQ^{-1} \begin{bmatrix} \frac{1}{\sqrt{\lambda_M}} \vr_M & \mRbar_M \end{bmatrix}
		\begin{bmatrix} 1 & \vzero^\top \\ \vzero & \mU \end{bmatrix},
	\end{equation}
	we only need to focus on the determinant of the middle term.
	There,
	\begin{align}
		\left| \det \begin{bmatrix} \frac{1}{\sqrt{\lambda_M}} \vr_M & \mRbar_M \end{bmatrix} \right|
		= \sqrt{\frac{\lambda_k}{\lambda_M}} \left| \det \begin{bmatrix} \frac{1}{\sqrt{\lambda_k}} \vr_k & \mRbar_k \end{bmatrix} \right|,
		\ \forall k,
		\nonumber
	\end{align}
	and the proof follows because $\lambda_k / \lambda_M \geq 1$ for any $k$.
\end{proof}

There are two points left to obtain the final algorithm.
First, $\mQ$ never changes throughout the algorithm and corresponds to a whitening of the input data.
It can by applied once and for all at the beginning of the algorithm.
Further multiplications are thus avoided.
Finally, we are only interested in $\vw$.
Being never needed, computation of $\mJ$ is omitted, which makes knowledge or estimation of $\mB$ moot.

\begin{figure}
	\centering
	\includegraphics[width=\linewidth]{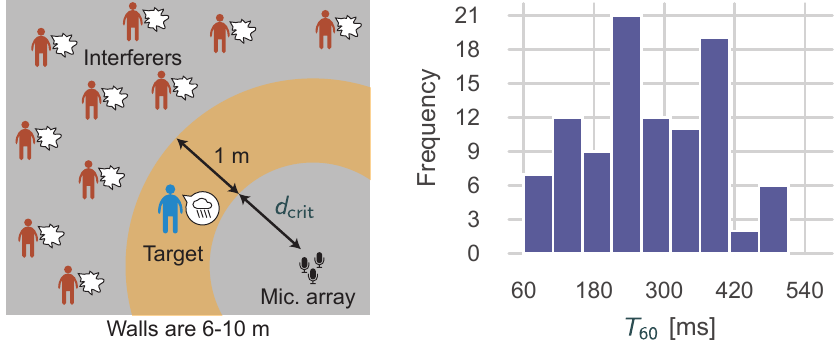}
	\caption{Right, the room setup for simulation. Left, the histogram of simulated reverberation times.}
	\flabel{room_setup}
\end{figure}

\begin{figure*}
	\centering
	\includegraphics[width=\linewidth]{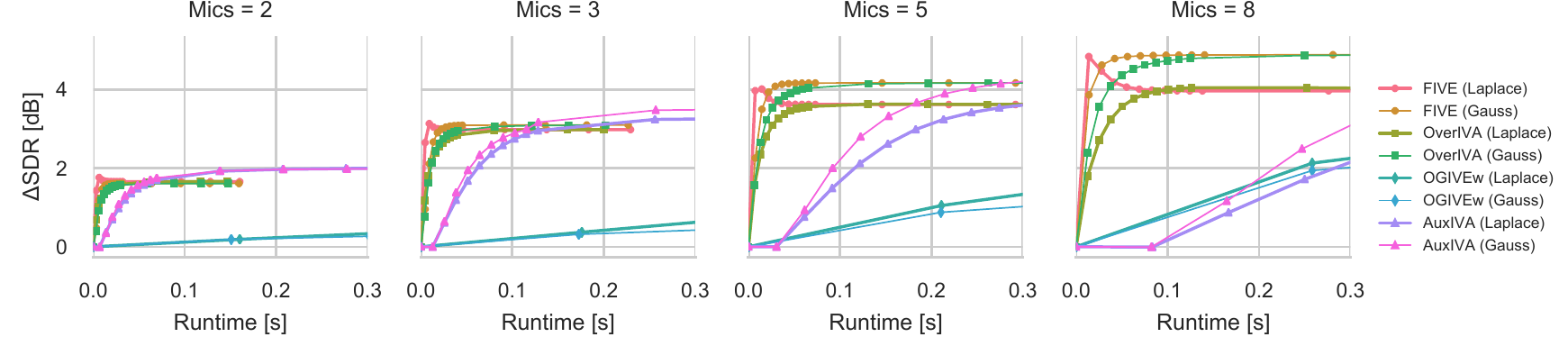}
	\caption{Mean convergence curves: SDR improvement as a function of the runtime for \SI{1}{\second} of input signal. From left to right, 2, 3, 5, and 8 microphones are used.}
	\flabel{conv_curve}
\end{figure*}

\section{Experiments}
\seclabel{experiments}

The performance of FIVE is assessed via simulations.
We study the convergence speed, the separation level after a few iterations, and the effect of mismatch in the background model.

\subsection{Experimental Setup}

We use the \texttt{pyroomacoustics} toolbox~\cite{Scheibler:2018di} to simulate 100 random rectangular rooms with walls between \SI{6}{\meter} and \SI{10}{\meter} and ceiling from \SI{2.8}{\meter} to \SI{4.5}{\meter} high.
Simulated reverberation times ($T_{60}$) range from \SI{60}{\milli\second} to \SI{540}{\milli\second}.
Sources and microphone array are placed at random at least \SI{50}{\centi\meter} away from the walls and between \SI{1}{\meter} and \SI{2}{\meter} high.
The array is circular and regular with 2, 3, 5, or 8 microphones, and radius such that neighboring elements are \SI{2}{\centi\meter} apart.
The distance from target source to array center is in $[d_{\text{crit}}, d_{\text{crit}}+1]$,  where the critical distance $d_{\text{crit}} = 0.057\sqrt{V/T_{60}}$, with $V$ the volume of the room~\cite{Kuttruff:2009uq}.
The $Q=10$ interferers are at least $d_{\text{crit}} + 1$ from the array.
We define $\SINR = \sigma_T^2/(Q \sigma_I^2 + \sigma_w^2)$, where $\sigma_T^2$ and $\sigma_I^2$ are the variance of target and interferers at the first microphone.
We fix $\SINR=$~\SI{5}{\decibel}.
The uncorrelated noise variance $\sigma_w^2$ is set to be \SI{1}{\percent} of the total noise-and-interference.
An illustration of the room setup and a histogram of the reverberation times are provided in~\ffref{room_setup}.
We use a 4096 points STFT with half-overlap and a Hamming window.

We compare FIVE to OverIVA~\cite{Scheibler:2019vx}, full AuxIVA~\cite{Ono:2011tn} with selection of the strongest output channel, and the gradient ascent based algorithm, OGIVE~\cite{Koldovsky:fn}.
The first three are run for 50 iterations, while the last one is run for 4000 iterations with step size 0.1, as specified in~\cite{Koldovsky:fn}.
We compare two source models: time-invariant Laplace and time-varying Gaussian.
Without going into details due to lack of space, these models lead to the weighting functions
\begin{equation}
	\varphi^{\text{Lap}}_n(r) = (2r)^{-1},\quad \text{and}\quad \varphi^{\text{Gau}}_n(r) = (r^2/F)^{-1},
\end{equation}
respectively.
The scale of the separated signals is restored by projection back onto the first microphone~\cite{Murata:2001gb}.
The separation performance is evaluated in terms of signal-to-distortion ratio (SDR) and signal-to-interference ratio (SIR)~\cite{Vincent:2006fz} using a popular toolbox~\cite{Raffel:2014uu}.

\begin{figure}
	\centering
	\includegraphics[width=\linewidth]{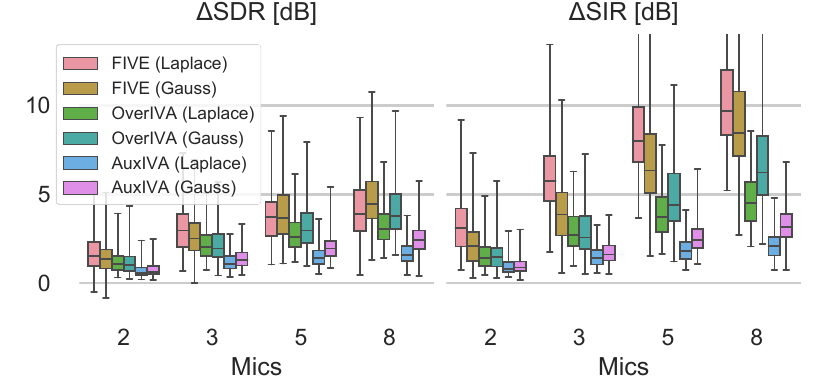}
	\caption{Box-plots of the SDR improvement after just three iterations.}
	\flabel{delta_sdr}
\end{figure}

\subsection{Convergence Speed}

\ffref{conv_curve} shows the mean evolution of SDR improvement ($\Delta$SDR) as a function of runtime.
The runtime is normalized per one second of input signal to gauge potential for real-time applications.
FIVE~(Laplace) is the fastest and reaches peak $\Delta$SDR in one to three iterations.
We observe however that the $\Delta$SDR subsequently decreases before reaching convergences.
Note that this is not a contradiction since the cost function~\eref{cost_function} is not the SDR.
We conjecture this to be due to a mismatch with the signal model.
In terms of speed, FIVE~(Gauss) is close behind but stably attains a larger $\Delta$SDR value.
OverIVA behaves similarly in terms of $\Delta$SDR, which is expected because it also minimizes~\eref{cost_function}.
Its convergence speed is about five times slower.
Doing full separation with AuxIVA improves $\Delta$SDR performance, likely due to a better modelling of the background as extra independent sound sources.
However, convergence is at least one order of magnitude slower, which hits particularly hard when using more microphones.
The gradient-based OGIVE converges at a much slower pace, but eventually reaches similar $\Delta$SDR values, although outside the limits of \ffref{conv_curve}.
We do admit, however, that its runtime might improve with a more careful implementation.

\subsection{Separation Performance after Three Iterations}

\begin{figure}
	\centering
  \includegraphics[width=\linewidth]{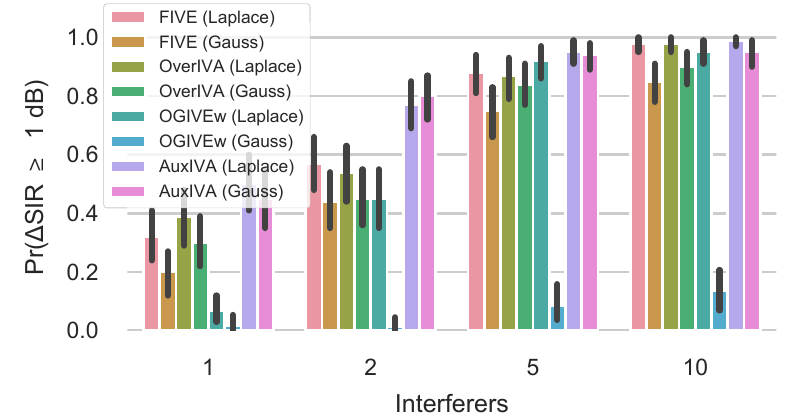}
	\caption{Success rate at $\SINR=$~\SI{0}{\decibel} and for different number of interferers.}
	\flabel{bg_mismatch}
\end{figure}

\ffref{delta_sdr} displays box-plots of the $\Delta$SDR and $\Delta$SIR of FIVE, OverIVA, and AuxIVA after three iterations.
We leave out OGIVE of the comparison because it was difficult to include it in a meaningful manner.
In all cases FIVE dominates, with OverIVA behind, and AuxIVA last.
Also recall that three iterations of AuxIVA is about ten times longer than the other two.
In general, time-varying Gauss model achieves higher $\Delta$SDR and Laplace model higher $\Delta$SIR.

\subsection{Effect of Background Model Mismatch}

In our experiments so far, the background has been composed of ten interference sources which is close to the Gaussian background assumption.
We now rerun the experiment with the conditions modified as follows.
The SINR is decreased to \SI{0}{\decibel} and the number of microphones set to three.
Then, we run the experiment with one, two, five, and ten interferers and measure the success rate of each algorithm.
The success is defined as $\Delta$SIR $\geq$ \SI{1}{\decibel}.
The experiment result is shown in \ffref{bg_mismatch}
When there is only one interferer, we expect all algorithms to fail because it is not possible to tell which source is the target without prior information.
Indeed, even AuxIVA is slightly lower than 0.5, meaning that it probably separates the sources, but picks the wrong one half of the time.
Other algorithms fail more often, which implies that separation itself fails.
As we increase the number of interferers and the background approach Gaussianity, the success rate of all algorithms increases, with the exception of OGIVE~(Gauss).
There is not much difference between FIVE and OverIVA, but AuxIVA, which does not assume a specific background model, performs markedly better.


\section{Conclusion}
\seclabel{conclusion}

We presented a deceptively simple algorithm for BSE which can be described as iterative maximization of the SINR.
The algorithm can be rigorously derived and its convergence is guaranteed.
In experiments we showed that the proposed algorithm is blazingly fast, only needing a few iterations, even for up to eight microphones.
In contrast, full IVA takes an order of magnitude longer, or more, to obtain the same SDR improvement.
However, our method assumes a Gaussian distributed background and its performance degrades, sometimes significantly, when this is not fulfilled.
Because BSE relies exclusively on the cost function for the extraction, a crucial next step is to identify a more suitable background model.
It should be flexible enough to accommodate a wide variety of conditions, yet offer good contrast between target and background, for example as in~\cite{Koldovsky:2018bw}.


\bibliographystyle{IEEEtran}
\bibliography{IEEEabrv,refs}

\begin{thebibliography}{10}
\providecommand{\url}[1]{#1}
\csname url@samestyle\endcsname
\providecommand{\newblock}{\relax}
\providecommand{\bibinfo}[2]{#2}
\providecommand{\BIBentrySTDinterwordspacing}{\spaceskip=0pt\relax}
\providecommand{\BIBentryALTinterwordstretchfactor}{4}
\providecommand{\BIBentryALTinterwordspacing}{\spaceskip=\fontdimen2\font plus
\BIBentryALTinterwordstretchfactor\fontdimen3\font minus
  \fontdimen4\font\relax}
\providecommand{\BIBforeignlanguage}[2]{{%
\expandafter\ifx\csname l@#1\endcsname\relax
\typeout{** WARNING: IEEEtran.bst: No hyphenation pattern has been}%
\typeout{** loaded for the language `#1'. Using the pattern for}%
\typeout{** the default language instead.}%
\else
\language=\csname l@#1\endcsname
\fi
#2}}
\providecommand{\BIBdecl}{\relax}
\BIBdecl

\bibitem{Statistics:1985tc}
P.~J. Huber, ``Projection pursuit,'' \emph{Ann. Stat.}, vol.~13, no.~2, pp.
  435--475, Jun. 1985.

\bibitem{Cardoso:1993ii}
J.~F. Cardoso and A.~Souloumiac, ``Blind beamforming for non-{G}aussian
  signals,'' \emph{IET}, vol. 140, no.~6, p. 362, 1993.

\bibitem{Comon:1512057}
P.~Comon and C.~Jutten, \emph{Handbook of blind source separation: independent
  component analysis and applications}, 1st~ed.\hskip 1em plus 0.5em minus
  0.4em\relax Oxford, UK: Academic Press/Elsevier, 2010.

\bibitem{Amari:1998cz}
S.~Amari and A.~Cichocki, ``Adaptive blind signal processing-neural network
  approaches,'' \emph{Proc. IEEE}, vol.~86, no.~10, pp. 2026--2048, 1998.

\bibitem{CrucesAlvarez:2004hq}
S.~A. Cruces-Alvarez, A.~Cichocki, and S.~Amari, ``From blind signal extraction
  to blind instantaneous signal separation: Criteria, algorithms, and
  stability,'' \emph{{IEEE} Trans. Neural Netw.}, vol.~15, no.~4, pp. 859--873,
  Jul. 2004.

\bibitem{Javidi:2010if}
S.~Javidi, D.~P. Mandic, and A.~Cichocki, ``Complex blind source extraction
  from noisy mixtures using second-order statistics,'' \emph{{IEEE} Trans.
  Circuits Syst. {I}}, vol.~57, no.~7, pp. 1404--1416, Jul. 2010.

\bibitem{Koldovsky:fn}
Z.~Koldovsk{\'{y}} and P.~Tichavsk{\'{y}}, ``Gradient algorithms for complex
  non-{G}aussian independent component/vector extraction, question of
  convergence,'' \emph{{IEEE} Trans. Signal Process.}, vol.~67, no.~4, pp.
  1050--1064, Dec. 2018.

\bibitem{Koldovsky:2018bw}
Z.~Koldovsk{\'{y}}, P.~Tichavsk{\'{y}}, and N.~Ono, ``Orthogonally-constrained
  extraction of independent non-{G}aussian component from non-{G}aussian
  background without {ICA},'' in \emph{Latent Variable Analysis and Signal
  Separation}.\hskip 1em plus 0.5em minus 0.4em\relax Cham: Springer, Cham,
  Jul. 2018, pp. 161--170.

\bibitem{Smaragdis:1998kl}
P.~Smaragdis, ``Blind separation of convolved mixtures in the frequency
  domain,'' \emph{Neurocomputing}, vol.~22, no. 1-3, pp. 21--34, Nov. 1998.

\bibitem{Sawada:fk}
H.~Sawada, S.~Araki, and S.~Makino, ``Measuring dependence of bin-wise
  separated signals for permutation alignment in frequency-domain {BSS},'' in
  \emph{Proc. IEEE ISCAS}, New Orleans, LA, USA, May 2007, pp. 3247--3250.

\bibitem{Hiroe:2006ib}
A.~Hiroe, ``Solution of permutation problem in frequency domain {ICA}, using
  multivariate probability density functions,'' in \emph{ASIACRYPT 2016}.\hskip
  1em plus 0.5em minus 0.4em\relax Berlin, Heidelberg: Springer Berlin
  Heidelberg, 2006, pp. 601--608.

\bibitem{Kim:2006ex}
T.~Kim, H.~T. Attias, S.-Y. Lee, and T.-W. Lee, ``Blind source separation
  exploiting higher-order frequency dependencies,'' \emph{IEEE Trans. Audio,
  Speech, Language Process.}, vol.~15, no.~1, pp. 70--79, Dec. 2006.

\bibitem{Ono:2011tn}
N.~Ono, ``Stable and fast update rules for independent vector analysis based on
  auxiliary function technique,'' in \emph{Proc. {IEEE} WASPAA}, New Paltz, NY,
  USA, Oct. 2011, pp. 189--192.

\bibitem{Cardoso:1994wj}
J.~F. Cardoso, ``On the performance of orthogonal source separation
  algorithms,'' in \emph{Proc. {IEEE} EUSIPCO}, Edinburgh, UK, Sep. 1994, pp.
  776--779.

\bibitem{VanTrees:2002ec}
H.~L. Van~Trees, \emph{Optimum Array Processing}.\hskip 1em plus 0.5em minus
  0.4em\relax New York, USA: John Wiley {\&} Sons, Inc., Mar. 2002.

\bibitem{Yeredor:hr}
A.~Yeredor, ``On hybrid exact-approximate joint diagonalization,'' in
  \emph{Proc. {IEEE} CAMSAP}, Dec. 2009, pp. 312--315.

\bibitem{Allen:1977in}
J.~Allen, ``Short term spectral analysis, synthesis, and modification by
  discrete {F}ourier transform,'' \emph{IEEE Trans. Acoust., Speech, Signal
  Process.}, vol.~25, no.~3, pp. 235--238, Jun. 1977.

\bibitem{Ono:2010hh}
N.~Ono and S.~Miyabe, ``Auxiliary-function-based independent component analysis
  for super-{G}aussian sources,'' \emph{Proc. LVA/ICA}, vol. 6365, no.~6, pp.
  165--172, Sep. 2010.

\bibitem{Lange:2016wp}
K.~Lange, \emph{{MM} optimization algorithms}.\hskip 1em plus 0.5em minus
  0.4em\relax SIAM, 2016.

\bibitem{Ono:2012wa}
N.~Ono, ``Fast stereo independent vector analysis and its implementation on
  mobile phone,'' in \emph{Proc. IWAENC}, Aachen, DE, Sep. 2012.

\bibitem{Scheibler:2018di}
R.~Scheibler, E.~Bezzam, and I.~Dokmani{\'c}, ``Pyroomacoustics: A {P}ython
  package for audio room simulations and array processing algorithms,'' in
  \emph{Proc. {IEEE} ICASSP}, Calgary, CA, Apr. 2018, pp. 351--355.

\bibitem{Kuttruff:2009uq}
H.~Kuttruff, \emph{Room acoustics}.\hskip 1em plus 0.5em minus 0.4em\relax CRC
  Press, 2009.

\bibitem{Scheibler:2019vx}
R.~Scheibler and N.~Ono, ``Independent vector analysis with more microphones
  than sources,'' in \emph{Proc. {IEEE} WASPAA}, New Paltz, NY, USA, Oct. 2019,
  accepted.

\bibitem{Murata:2001gb}
N.~Murata, S.~Ikeda, and A.~Ziehe, ``An approach to blind source separation
  based on temporal structure of speech signals,'' \emph{Neurocomputing},
  vol.~41, no. 1-4, pp. 1--24, Oct. 2001.

\bibitem{Vincent:2006fz}
E.~Vincent, R.~Gribonval, and C.~Fevotte, ``Performance measurement in blind
  audio source separation,'' \emph{IEEE Trans. Audio, Speech, Language
  Process.}, vol.~14, no.~4, pp. 1462--1469, Jun. 2006.

\bibitem{Raffel:2014uu}
C.~Raffel, B.~McFee, E.~J. Humphrey, J.~Salomon, O.~Nieto, D.~Liang, D.~P.~W.
  Ellis, C.~C. Raffel, B.~Mcfee, and E.~J. Humphrey, ``mir\_eval: A transparent
  implementation of common {MIR} metrics,'' in \emph{Proc. ISMIR}, 2014.

\end{thebibliography}

\end{document}